\documentclass[11pt]{article}
\usepackage[
	includefoot,
	margin=3cm,
	a4paper,
]{geometry}
\usepackage{microtype}
\usepackage[utf8]{inputenc}
\usepackage{enumerate}
\usepackage{amsmath}
\usepackage{amsfonts} 
\usepackage{nicefrac}
\usepackage{amssymb}
\usepackage{amsthm}
\usepackage{mathtools}
\usepackage{tikz}
\usetikzlibrary{calc}
\tikzstyle{para}=[rectangle,draw=black,minimum height=.8cm,fill=gray!10,rounded corners=1mm, on grid]
\usepackage{pgfplots}
\pgfplotsset{compat=newest}
\usepackage{pgfplotstable}
\usepackage{xcolor}
\usepackage{authblk}

\tikzstyle{path} = [color=black!10,line cap=round, line join=round, line width=12pt]
\tikzset{
	colornode/.style = {
		circle,
		draw=#1!70!black,
		very thick,
		fill=#1
	}
}

\definecolor{r}{rgb}{1.0, 0.4, 0.4}
\definecolor{r0}{rgb}{1.0, 0.7, 0.4}
\definecolor{r1}{rgb}{1.0, 0.4, 0.0}
\definecolor{r2}{rgb}{0.8, 0, 0.4}
\definecolor{r3}{rgb}{1, 0.4, 0.3}
\definecolor{b}{rgb}{0.4, 0.4, 1.0}
\definecolor{b0}{rgb}{0.4, 0.7, 1.0}
\definecolor{b1}{rgb}{0.0, 0.4, 1.0}
\definecolor{b2}{rgb}{0.4, 0.0, 0.8}
\definecolor{g}{rgb}{0.2, 0.9, 0.3}
\definecolor{y2}{rgb}{1, 1, 0.3}

\DeclareRobustCommand{\tikzdot}[1]{\tikz[baseline=-0.6ex]{\node[draw,fill=#1,inner sep=2pt,circle] at (0,0) {};}}

\usepackage{tabularx}
\usepackage{booktabs}
\usepackage{paralist}

\usepackage{algorithm2e}
\DontPrintSemicolon

\usepackage[pdfdisplaydoctitle,menucolor=orange!40!black,filecolor=magenta!40!black,urlcolor=blue!40!black,linkcolor=red!40!black,citecolor=green!40!black,colorlinks]{hyperref}

\newcommand{\ExternalLink}{%
    \tikz[color=magenta, x=1.2ex, y=1.2ex, baseline=-0.05ex]{%
        \begin{scope}[x=1ex, y=1ex]
            \clip (-0.1,-0.1) 
                --++ (-0, 1.2) 
                --++ (0.6, 0) 
                --++ (0, -0.6) 
                --++ (0.6, 0) 
                --++ (0, -1);
            \path[draw, 
                line width = 0.5, 
                rounded corners=0.5] 
                (0,0) rectangle (1,1);
        \end{scope}
        \path[draw, line width = 0.5] (0.5, 0.5) 
            -- (1, 1);
        \path[draw, line width = 0.5] (0.6, 1) 
            -- (1, 1) -- (1, 0.6);
        }
}

\DeclareUnicodeCharacter{1EF3}{\`y}
\usepackage[
	backend=biber,
	isbn=false,
	url=false,
	doi=true,
	backref=true,
	hyperref=true,
	natbib=true,
	maxcitenames=3,
	maxbibnames=99,
	sortcites]{biblatex}
\renewbibmacro{in:}{\ifentrytype{article}{}{\printtext{\bibstring{in}\intitlepunct}}}
\newbibmacro{doilink}[1]{%
	\iffieldundef{doi}{%
		\iffieldundef{url}{%
			#1%
        }{%
		\href{\thefield{url}}{$\ExternalLink$}%
		}%
	}{%
		\href{https://doi.org/\thefield{doi}}{$\ExternalLink$}%
	}%
}
\DeclareFieldFormat*{doi}{\usebibmacro{doilink}{#1}}

\usepackage[nameinlink,sort&compress,capitalize,noabbrev]{cleveref}

\usepackage[backgroundcolor=gray!10,textsize=footnotesize]{todonotes}

\newtheorem{theorem}{Theorem}

\newtheorem{observation}[theorem]{Observation}
\newtheorem{corollary}[theorem]{Corollary}
\newtheorem{proposition}[theorem]{Proposition}

\theoremstyle{definition}
\newtheorem{definition}[theorem]{Definition}

\crefname{rrule}{Rule}{Rules}
\crefname{figure}{Figure}{Figures}

\newcommand{\prob}[1]{\textnormal{\textsc{#1}}}

\newcommand{\myproblem}[5]{
	\begin{center}
	\begin{minipage}{0.95\columnwidth}
		\noindent
		\prob{#1}
		\vspace{5pt}\\
		\setlength{\tabcolsep}{3pt}
		\begin{tabularx}{\textwidth}{@{}lX@{}}
			\textbf{#2}     & #3 \\
			\textbf{#4}  & #5
		\end{tabularx}
	\end{minipage}
	\end{center}
}

\newcommand{\problemdef}[3]{\myproblem{#1}{Input:}{#2}{Question:}{#3}}

\DeclarePairedDelimiterX{\abs}[1]{\lvert}{\rvert}{#1}
\DeclarePairedDelimiterX{\norm}[1]{\lVert}{\rVert}{#1}
\DeclarePairedDelimiterX{\ceil}[1]{\lceil}{\rceil}{#1}

\newcommand{\N}{\mathbb{N}}
\newcommand{\Nzero}{\N_0}

\newcommand{\calF}{\mathcal{F}}

\newcommand{\cocl}[1]{\ensuremath{\operatorname{#1}}}
\newcommand{\W}[1]{\cocl{W[#1]}}
\newcommand{\Wone}{\W{1}}

\newcommand{\NP}{\cocl{NP}}

\newcommand{\coNP}{\cocl{coNP}}
\newcommand{\XP}{\cocl{XP}}
\newcommand{\poly}{\cocl{poly}}

\newcommand{\NPincoNPslashpoly}{\ensuremath{\NP\subseteq \coNP/\poly}}

\newcommand{\bigO}{\mathcal{O}}
\newcommand{\yes}{\textnormal{\texttt{yes}}}
\newcommand{\no}{\textnormal{\texttt{no}}}

\newcommand{\NN}{\mathbb{N}}

\newcommand{\FF}{\mathbb{F}}

\newcommand{\AAA}{\mathcal{A}}

\newcommand{\III}{\mathcal{I}}

\newcommand{\SSS}{\mathcal{S}}

\DeclareMathOperator{\rep}{rep}

\newcommand{\repr}{\ensuremath{\subseteq_{\rep}}}

\DeclareMathOperator{\dist}{dist}

\newcommand{\smallbinom}[2]{\Bigl(\begin{array}{@{}c@{}}#1\\#2\end{array}\Bigr)}

\newcommand{\oneto}[1]{[ #1 ]} 

\newcommand{\ie}{i.\,e.,\ }
\newcommand{\eg}{e.g.\ }

\newcommand{\fsp}{\textnormal{\textsc{\Fair{} Shortest Path}}}

\newcommand{\abX}[3]{\ensuremath{#1}-\ensuremath{#2}\nobreakdash-#3}
\newcommand{\abpath}[2]{\abX{#1}{#2}{path}}
\newcommand{\abpaths}[2]{\abX{#1}{#2}{paths}}
\newcommand{\stpath}{\abpath{s}{t}}
\newcommand{\stpaths}{\abpaths{s}{t}}

\newcommand{\tworows}[2]{\begin{tabular}{c}{#1}\\{#2}\end{tabular}}
\newcommand{\threerows}[3]{\begin{tabular}{c}{#1}\\{#2}\\{#3}\end{tabular}}
\newcommand{\distto}[1]{\tworows{Distance to}{#1}}
\newcommand{\disttoc}[2]{\threerows{Distance to}{#1}{#2}}

\newcommand*{\defeq}{\mathrel{\vcenter{\baselineskip0.5ex \lineskiplimit0pt\hbox{\scriptsize.}\hbox{\scriptsize.}}}=}

\title{\Large \bf The Structural Complexity Landscape of Finding Balance-Fair~Shortest~Paths}

\author[1,2]{Matthias Bentert}
\author[1]{Leon Kellerhals}
\author[1]{Rolf Niedermeier}

\affil[1]{\small
	Technische Universit\"at Berlin
}

\affil[2]{\small
	University of Bergen\protect\\
	\texttt{\textrm{matthias.bentert@uib.no}}\protect\\
	\texttt{\textrm{leon.kellerhals@tu-berlin.de}}%
}

\date{}

\bibliography{strings-long,bibliography}

\newcommand{\fair}{balance-fair}
\newcommand{\Fair}{Balance-fair}

\newcommand{\col}{\ensuremath{\chi}}
\newcommand{\colors}{\ensuremath{\oneto{c}}}
\newcommand{\ncols}[1][]{\ensuremath{c}}

\begin{document}

\maketitle

\begin{abstract}
We study the parameterized complexity of finding shortest~$s$-$t$-paths with an additional fairness requirement.
The task is to compute a shortest path in a vertex-colored graph where each color appears (roughly) equally often in the solution.
We provide a complete picture of the parameterized complexity landscape of the problem with respect to structural parameters by showing a tetrachotomy including polynomial kernels, fixed-parameter tractability, XP-time algorithms (and \Wone-hardness), and para-NP-hardness. 
\end{abstract}

\section{Introduction}
Computing shortest paths in graphs is a fundamental problem in algorithmic graph theory.
In many applications, however, not every shortest path is equally valuable.
Consider a group of people traveling together.
Some members of the group may prefer traveling next to rivers, others prefer forests, and still others may prefer urban regions.
To satisfy all group members' preferences and not create any envy, the group decides to search for a \emph{fair} route, that is, a route that (roughly) encounters the same amount of river sights, forests, and urban regions.
Clearly, such a travel from a starting point~$s$ to an endpoint~$t$ can be modeled as finding an~\stpath{} in a graph.
To model the fairness aspect, the vertices (or edges\footnote{We decided for the vertex variant, but both variants are computationally equivalent.}) are colored, say blue vertices for paths next to rivers, green for forest routes, and gray for streets in a city.
The task is to find a \emph{\fair{}} path, i.e., a path in which every color appears equally often.
We provide an example with three colors in \cref{fig:introexample}.
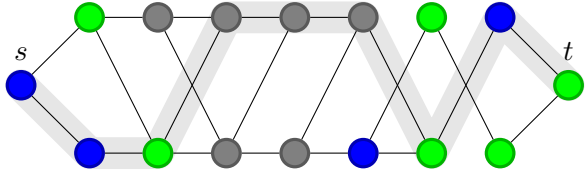
\begin{figure}[t]
\centering
\begin{tikzpicture}[scale=.9]
		\draw[path] (-4.5,0) -- (-3.5,-1) -- (-2.5,-1) -- (-1.5,1) -- (-.5,1) -- (.5,1) -- (1.5,-1) -- (2.5,1) -- (3.5,0);

		\node[colornode=blue, label=$s$] at (-4.5,0) (s) {};
		\node[colornode=green, label=$t$] at(3.5,0) (t) {};
                                                      ,
		\node[colornode=blue] at (-3.5,-1) (p11) {} edge(s);
		\node[colornode=green] at (-2.5,-1) (p12) {} edge(p11);
		\node[colornode=gray] at (-1.5,-1) (p13) {} edge(p12);
		\node[colornode=gray] at (-.5,-1) (p14) {} edge(p13);
		\node[colornode=blue] at (.5,-1) (p15) {} edge(p14);
		\node[colornode=green] at (1.5,-1) (p16) {} edge(p15);
		\node[colornode=green] at (2.5,-1) (p17) {} edge (t);
		
		\node[colornode=green] at (-3.5,1) (p21) {} edge(s) edge(p12);
		\node[colornode=gray] at (-2.5,1) (p22) {} edge(p21)  edge(p13);
		\node[colornode=gray] at (-1.5,1) (p23) {} edge(p22) edge(p12);
		\node[colornode=gray] at (-.5,1) (p24) {} edge(p23) edge(p13);
		\node[colornode=gray] at (.5,1) (p25) {} edge(p14) edge(p16) edge(p24);
		\node[colornode=green] at (1.5,1) (p26) {} edge(p15) edge(p17);
		\node[colornode=blue] at (2.5,1) (p27) {} edge(p16) edge (t);
		
	\end{tikzpicture}
	\caption{A graph with colored vertices (blue, green, and gray). The highlighted path is a shortest path between~$s$ and~$t$ and contains three vertices of each color. Thus, it is \fair{}.}
	\label{fig:introexample}	
\end{figure}%

We study the parameterized complexity landscape of finding \fair{} shortest \stpaths.

\paragraph{Related work.}
Path finding in vertex-colored graphs has been a subject of broad and 
intensive study. Here, we only point to algorithmically motivated work that seems particularly close
to our scenario of \fair{} shortest paths.
First of all, finding colorful paths (each 
color appears at most once) is an important algorithmic 
topic, both in static and in temporal graphs~\cite{alon1995color,DH21}.

\citet{CIMTP21} analyze the complexity of finding 
tropical paths (shortest and longest), that is, paths containing
each color in the input graph at least once.
Another close (and also broad) area is that of finding 
resource-constrained shortest paths~\cite{PG13}, where roughly speaking 
the desired path shall have minimum cost and only a limited consumption
of resources. Generally, the problem is NP-hard~\cite{HZ80} and 
received a lot of attention in recent years \cite{ATHK21,ford2022backtracking,Irnich2005}.
Note, however, that there are no immediate fairness aspects modeled here.
\citet{hanaka2021computing} are somewhat closer to fairness aspects in path finding.
They study shortest paths under diversity aspects,
meaning that they search for multiple shortest paths that 
are maximally different from each other. This fits into the recent trend of
finding diverse sets of 
solutions to optimization problems~\cite{BFJMOPR20,PT19}. 
Coming back to \fair{} paths, we have shown in previous work that finding \fair{} paths of length at most some given~$\ell$ is fixed-parameter tractable by~$\ell$ and \Wone-hard parameterized by the number of colors in the input graph \cite{BKN22}.
Note that searching for \fair{} shortest \stpaths{} is a special case where~$\ell$~is the distance between the two terminals~$s$ and~$t$ in the input graph. 

Finally, we only mention in passing that fairness aspects are currently investigated in all kinds of optimization problems (particularly graph-based ones), including topics such as graph-based data  clustering~\cite{AEKMMPVW20,AEKM20,FM21,FKN21}, influence maximization~\cite{KRBHJ20}, and graph mining~\cite{KT21,DMCL22}.
Fairness notions similar to \fair{ness} are studied in various contexts~\cite{AM17}.
A slightly more general scenario is considered for the classic graph problems \textsc{Vertex Cover} and \textsc{Edge Cover}~\cite{BBB21}.
Particularly close to our setting is the work by \citet{CKLV19}.
They study \fair{ness} (their definition only differs in small technical details from ours) in problems that can be solved by finding a largest independent set in a matroid or the intersection of two matroids.
Their results \eg imply that finding \fair{} spanning trees can be solved in polynomial time when the number of colors is constant.\footnote{By a result of \citet{brezovec1988matroid}, the polynomial-time solvability also holds for an unbounded number of colors.}

\paragraph{Our contributions.}
Altogether, we provide a quite complete picture of the parameterized
complexity landscape of a well-motivated fairness scenario for 
one of the back-bone problems in network algorithmics: finding 
shortest paths.
Our contribution is a complete tetrachotomy (a division into four parts) for structural parameterizations for \fsp{} (see \cref{sec:prelim} for a formal definition) depicted in \cref{parameterHierarchy}.
\begin{figure}[t]
\centering
\begin{tikzpicture}[node distance=2*0.60cm and 3.7*0.50cm, every node/.style={scale=0.65}]
\tikzstyle{pk}=[fill=green!30]
\tikzstyle{nopk}=[fill=yellow!50]
\tikzstyle{wone}=[fill=orange!30]
\tikzstyle{nph}=[fill=red!30]
\tikzstyle{res}=[very thick]
\linespread{1}
\node[para,pk] (vc) {Minimum Vertex Cover};
\node[para, xshift=4.1cm,nopk,res] (ml) [right=of vc] {\tworows{Max Leaf \#*}{\small [Prop.\;\ref{prop:pkfen}]}};
\node[para, xshift=-3.4cm,pk] (dc) [left=of vc] {\distto{Clique}};

\node[para, yshift=-5mm,xshift=-14mm,nopk,res] (mcc) [below left=of dc] {\threerows{Minimum}{Clique Cover}{\small [Prop.\;\ref{prop:nokerntd}]}}
edge (dc);
\node[para,yshift=-7mm,xshift=-13.0mm,pk] (dcc) [below= of dc] {\distto{Co-Cluster}}
edge (dc)
edge[bend left=20] (vc);
\node[para,xshift=13.8mm,yshift=-5mm,pk] (dcl) [below= of dc] {\distto{Cluster}}
edge (dc)
edge (vc);
\node[para, xshift=14.7mm,yshift=-5mm,wone] (ddp) [below=of vc] {\distto{disjoint Paths \cite{BKN22}}}
edge (vc)
edge (ml);
\node[para,xshift=2mm,yshift=-5mm,nopk,res] (fes) [below =of ml] {\threerows{Feedback}{Edge Set}{\small [Obs.\;\ref{obs:fen}]}}
edge (ml); 
\node[para,xshift=-1mm,yshift=-5mm,nph,res] (bw) [below right=of ml] {\tworows{Bandwidth}{\small [Prop.\;\ref{prop:bipartite}]}}
edge (ml);
\node[para,xshift=-19.5mm,yshift=-06mm,below=of vc,pk,res] (nd) {\threerows{Neighborhood}{Diversity}{\small [Prop.\;\ref{prop:kernd}]}}
edge (vc)
edge (dc);

\node[para,xshift=2mm,nopk] (is) [below=of mcc] {\tworows{Maximum}{Independent Set}}
edge (mcc);
\node[para,xshift=5mm,yshift=1mm,pk,res] (dcg) [below= of dcc] {\disttoc{Cograph}{\small [Prop.\;\ref{prop:kercograph}]}}
edge (dcc)
edge (dcl);
\node[para,xshift=5mm,yshift=-1mm,nph,res] (dig) [below= of dcl] {\disttoc{Interval}{\small [Prop.\;\ref{prop:interval}]}}
edge (dcl)
edge [bend right=15] (ddp);
\node[para,fill=orange,yshift=-10mm,wone,res] (fvs) [below= of ddp] {\threerows{Feedback}{Vertex Set}{\small [Obs.\;\ref{obs:fvn}]}}
edge (ddp)
edge (fes);
\node[para, xshift=3.0mm, yshift=-0mm,nopk,res] (td) [right=of ddp] {\tworows{Treedepth}{\small [Prop.\;\ref{prop:nokerntd}]}}
edge [bend right=28] (vc);
\node[para,nph] (mxd) [below= of bw] {\tworows{Maximum}{Degree}}
edge (bw);

\node[para,nopk] (ds) [below=of is] {\tworows{Minimum}{Dominating Set}}
edge (is);
\node[para, yshift=-55mm, xshift=-48mm,nph,res] (dbp) [below left= of fvs] {\disttoc{Bipartite}{\small [Prop.\;\ref{prop:bipartite}]}}
edge (fvs);
\node[para, xshift=0mm, yshift=-7mm,nph,res] (dop) [below= of fvs] {\disttoc{Outerplanar}{\small [Prop.\;\ref{prop:bipartite}]}}
edge (fvs);
\node[para, yshift=-13mm,nph] (pw) [below= of td] {Pathwidth}
edge (ddp)
edge (td)
edge (bw);
\node[para,yshift=-10mm,nph] (hid) [below= of mxd] {$h$-index}
edge (ddp)
edge (mxd);
\node[para, nph,res,xshift=2mm] (gen) [left= of hid] {\tworows{Genus}{\small [Prop.\;\ref{prop:bipartite}]}}
edge (fes);

\node[para,xshift=1.5mm,nopk] (mxdia) [below=of ds] {\tworows{Max Diameter}{of Components \cite{BKN22}}}
edge (ds)
edge[bend right=5] (td)
edge[bend right=37] (nd);
\draw (mxdia.north east) to (dcg);
\node[para, yshift=5mm,nph] (tw) [below right= of dop] {Treewidth}
edge (dop)
edge (pw);
\node[para, yshift=-15mm,nph] (acn) [below = of gen] {\tworows{Acyclic}{Chromatic \#}}
edge (hid)
edge (gen)
edge (tw);
\node[para, yshift=-2.5mm,nph] (dpl) [below = of dop] {\distto{Planar}}
edge (dop)
edge (acn);
\node[para, xshift=-15mm, yshift=2mm,nph] (clw) [below= of dpl] {Clique-width}
edge (dcg)
edge[bend right=15] (tw);
\node[para,nph,res] (avgdist) [below=of mxdia] {\threerows{Average}{Distance}{\small [Obs.\;\ref{obs:dist}]}}
edge (mxdia);
\node[para, xshift=10mm, yshift=0mm,nph] (dch) [right= of avgdist] {\distto{Chordal}}
edge (dig)
edge (fvs);
\node[para,nph] (deg) [below= of acn] {Degeneracy}
edge (acn);
\node[para, yshift=-12mm, xshift=13mm,nph] (box) [left=of clw] {Boxicity}
edge (dig)
edge (nd)
edge (acn);
\node[para,nph] (cn) [below =of dbp] {Chromatic \#}
edge (deg)
edge (dbp);
\node[para, yshift=0mm, xshift=-4mm,nph] (dpf) [left= of cn] {\distto{Perfect}}
edge (dch)
edge (dcg)
edge (dbp);
\node[para,nph] (avd) [below=of deg] {\tworows{Average}{Degree}}
edge (deg);

\node[para,nph] (mnd) [below=of avd] {\tworows{Minimum}{Degree}}
edge (avd);

\node[para,nph] (mc) [below=of cn] {\tworows{Maximum}{Clique}}
edge (cn);
\node[para, yshift=5mm, xshift= 22mm,nph] (cho) [right= of mc] {Chordality}
edge (box)
edge (dch)
edge (cn)
edge (dcg);
\node[para,nph] (gi) [below=of avgdist] {Girth}
edge (avgdist);
\end{tikzpicture}
\caption{
The relations between structural graph parameters and our respective results for \fsp{}.
An edge from a parameter~$\alpha$ to a parameter~$\beta$ below~$\alpha$ means that~$\alpha$ upper-bounds~$\beta$ (the bounds are usually linear functions; see also \cite{Sch19}).
A parameter~$k$ is marked
green (\tikzdot{green!30}) if \fsp{} admits a polynomial kernel with~$k$,
yellow (\tikzdot{yellow!50}) if it is FPT with~$k$ but presumably does not admit a polynomial kernel,
orange~(\tikzdot{orange!30}) if the respective parameterized problem is in XP and \Wone-hard, and
red (\tikzdot{red!30}) if \fsp{} is \NP-hard for constant~$k$.
Parameters without a thick border obtain their classification from parameters above or below.
}
\label{parameterHierarchy}
\end{figure}

We divide the parameters into
\begin{inparaenum}[(i)]
	\item parameters allowing for polynomial kernels,
	\item parameters that allow for FPT-time algorithms but (presumably) not for polynomial kernels,
	\item parameters for which the resulting parameterized problem is \Wone-hard but allows for XP-time algorithms, and
	\item parameters for which the resulting parameterized problem is para-NP-hard.
\end{inparaenum}

The remainder of this work is structured as follows.
We start with some preliminaries and basic definitions in \cref{sec:prelim}.
These include a formal problem definition and an overview over the parameters and tools used throughout the paper.
We then list in \cref{sec:obs} two results from our previous work \cite{BKN22} and derive observations that are relevant for our tetrachotomy.
Afterwards, we present all our results related to polynomial kernels in \cref{sec:kernel}.
We complete our tetrachotomy with two para-NP-hardness results in \cref{sec:paraNP} and conclude the paper in \cref{sec:concl}.

\section{Preliminaries}
\label{sec:prelim}
Let~$\N$ be the set of positive integers.
For~$n \in \N$, let~${\oneto{n} \defeq \{1, 2, \dots, n\}}$.

\paragraph{Graph theory.}
We use standard graph-theoretic terminology.
In particular, for an undirected graph~${G = (V, E)}$ we set~$n \defeq \abs{V}$ and~$m \defeq \abs{E}$.
For a subset~$V' \subseteq V$ of the vertices, we use~$G[V'] \defeq (V', \{e \in E \mid e \subseteq V'\})$ to denote the subgraph of~$G$ induced by~$V'$ and denote by~$G - V'$ the subgraph~$G[V \setminus V']$.
The \emph{degree}~$\deg_G(v)$ of $v$ is the number of vertices the vertex~$v$ is adjacent to in~$G$.
A \emph{path} $P$ \emph{on~$\ell$~vertices} is a graph with vertex 
set~$\{v_1, v_2, \ldots, v_{\ell}\}$ and edge 
set~${\{ \{v_i, v_{i+1}\} \mid i \in \oneto{\ell-1}\}}$.
The vertices~$v_1$ and~$v_{\ell}$ are called \emph{endpoints}.
The \emph{length} of the path is the number of edges it contains.
A \emph{connected component} in a graph is a maximal set~$V'$ of vertices such that for each pair~$u,v \in V'$ of vertices in~$V'$, there is a path in the graph with endpoints~$u$ and~$v$.

\paragraph{Problem definition.}
Let~$G = (V, E)$ be a graph with two vertices~$s$ and~$t$.
An \stpath{}~$P$ is a path in $G$ whose endpoints are~$s$ and~$t$.
We denote by~$\dist_G(s, t)$ the length of a shortest \stpath{} in $G$.
Whenever clear from context, we may drop the subscript~$G$.
Let~$\chi \colon V \to \oneto{c}$ be a vertex coloring.
For a color~$i \in \oneto{c}$, we denote by~$\chi^i \defeq \{v \in V \mid \chi(v)=i\}$ the set of vertices of color~$i$.
For a subgraph~$H$ of~$G$ and a color~$i \in \oneto{c}$, we denote by~$\chi_H$ the coloring~$\chi$ restricted to the vertices of~$H$ and by~$\chi_H^i$ the set of vertices of color~$i$ in~$H$.
We say that a path~$P$ in~$G$ is \emph{\fair} if ${\max_{i \in \oneto{c}} \abs{\chi_P^i} - \min_{i \in \oneto{c}} \abs{\chi_P^i} \le 1}$.
The problem we study in this work is defined as follows.

\problemdef{\prob{\fsp}}
{An undirected graph~$G = (V, E)$, a vertex coloring~${\chi \colon V \to \oneto{c}}$, and two vertices~$s, t \in V$.}
{Is there a \fair{} shortest~\stpath~$P$ in~$G$?}

Observe that in the introduction we motivated the task of finding a solution path containing \emph{exactly} the same number of vertices of each color.
However, if the number of vertices in a shortest \stpath{} is divisible by the number~$c$ of colors without remainder, then any \fair{} shortest \stpath{} also satisfies this stronger requirement.
We have decided for the definition above as it is a more interesting problem in the case when the number of vertices in a solution is not divisible by~$c$ without remainder.

\paragraph{Matroids and representative families.}
A pair~$M = (U, \III)$, where~$U$ is called \emph{ground set} and $\III$ is a family of subsets (called \emph{independent sets}) of~$U$,
is a \emph{matroid} if
\begin{inparaenum}[(i)]
	\item $\emptyset \in \III$,
	\item if~${A' \subseteq A}$ and~${A \in \III}$, then~${A' \in \III}$ (hereditary property), and
	\item if~${A, B \in \III}$ and~${\abs{A} < \abs{B}}$, then there is an~${e \in B \setminus A}$ such that~${A \cup \{e\} \in \mathcal I}$ (exchange property).
\end{inparaenum}
An inclusion-wise maximal independent set is a \emph{basis} of~$M$.
It follows from the exchange property that all bases of~$M$ have the same size.
This size is called the \emph{rank} of~$M$.
Let~$A$ be a matrix over a finite field~$\FF$, and let~$U$ be the set of columns of $A$.
We associate a matroid~$M = (U, \III)$ with $A$ as follows.
A set~$X \subseteq U$ is independent (i.e., $X \in \III$) if the columns in $X$ are linearly independent over $\FF$.
We say that the matroid is \emph{linear} and that~$A$ represents $M$.

We use matroids in order to compute representative families which are defined as follows.
 
\begin{definition}
	\label{def:repr}
	Given a matroid~$M = (U, \III)$ and a family~$\AAA$ of size-$p$~subsets of~$U$,
	we say that a subfamily~$\widehat\AAA$ is a~\emph{$q$-representative} for~$\AAA$ if the following holds:
	For every set~$B \subseteq U$ of size at most~$q$, if there is a set~$A \in \AAA$ with~$A \cap B = \emptyset$ and~$A \cup B \in \III$,
	then there is a set~$\widehat A \in \widehat \AAA$ with~$\widehat A \cap B = \emptyset$ and~$\widehat A \cup B \in \III$.
	We then write~$\widehat \AAA \repr^q \AAA$.
\end{definition}

\citet{fomin2016representative} gave multiple efficient recipes for the computation of representative families.
The following works for linear matroids.
Here, $\omega < 2.373$ is the matrix multiplication constant~\cite{matrixmultiplication}.

\begin{theorem}[\citet{fomin2016representative}]
	\label{thm:repr}
	Let~$M = (U, \III)$ be a linear matroid of rank~$p+q=k$ given together with its representation matrix~$A$ over a field~$\FF$.
	Let~$\AAA = \{A_1, \dots, A_t\}$ be a family of independent sets of size~$p$ and let~$w \colon \AAA \to \Nzero$ be a weight function.
	Then a $q$-representative family~$\widehat\AAA \subseteq \AAA$ for~$\AAA$ with at most~$\binom{p+q}{p}$ sets can be found in
	\[
		\bigO\Big(\smallbinom{p+q}{p}tp^{\omega} + t\smallbinom{p+q}{q}^{\omega-1}\Big)
	\]
	operations over~$\FF$.
\end{theorem}

\paragraph{Parameterized complexity.}
Let~$\Sigma$ be a finite alphabet.
A \emph{parameterized problem} is a set of instances~$(I, k)$ where $I \in \Sigma^*$ is a problem instance from some finite alphabet~$\Sigma$ and~$k \in \N$ is the \emph{parameter}.
A parameterized problem~$L$ is \emph{fixed-parameter tractable} if $(I, k) \in L$ can be decided in~$f(k) \cdot \abs{I}^{\bigO(1)}$ time, where~$f$ is a computable function only depending on~$k$.
We call~$(I, k)$ a \emph{\yes-instance} (of~$L$) if~$(I, k) \in L$.
The class \XP{} contains all parameterized problems which can be decided in polynomial time if the parameter~$k$ is constant, that is, in $\abs{I}^{f(k)}$ time.
To show that a parameterized problem $L$ is presumably not fixed-parameter tractable, one may use a \emph{parameterized reduction} from a \Wone-hard problem to~$L$~\cite{DF13}.
A parameterized reduction from a parameterized problem $L$ to another parameterized problem~$L'$ is an algorithm satisfying the following.
There are two computable functions~$f$ and~$g$, such that given an instance~$(I, k)$ of~$L$, the algorithm computes in~$f(k) \cdot \abs{I}^{\bigO(1)}$~time an instance~$(I', k')$ of $L'$ such that $(I, k)$ is a \yes-instance if and only if $(I', k')$ is a \yes-instance and~$k' \leq g(k)$.
A \emph{kernelization} is an algorithm that, given an instance~$(I, k)$ of a parameterized problem~$L$, computes in~$|I|^{\bigO(1)}$ time an instance~$(I', k')$ of~$L$ (the \emph{kernel}) such that~$(I, k)$ is a \yes-instance if and only if~$(I', k')$ is a \yes-instance and~$\abs{I'}+k' \le f(k)$ for some computable function~$f$ only depending on~$k$.
We say that~$f$ measures the \emph{size} of the kernel.
If~$f$ is a polynomial, then we say that~$P$ admits a \emph{polynomial kernel}.
While all fixed-parameter tractable problems admit a kernel, it is not clear whether all of them admit one of polynomial size.
Indeed, assuming that \NPincoNPslashpoly, one can show that certain parameterized problems do not admit a polynomial kernel.
This can be done \eg via an OR-cross-composition or via a polynomial parameter transformation from a parameterized problem that (presumably) does not admit a polynomial kernel.
For the definition of OR-cross-compositions, we first need the following.
Given an \NP-hard problem~$L$, an equivalence relation~$\mathcal R$ on the instances of~$L$ is a \emph{polynomial equivalence relation} if
\begin{inparaenum}[(i)]
	\item one can decide for any two instances in polynomial time whether they belong to the same equivalence class, and
	\item for any finite set $S$ of instances, $\mathcal R$ partitions the set into at most $(\max_{I \in S} \abs{I})^{\bigO(1)}$ equivalence classes.
\end{inparaenum}

\begin{definition}[OR-cross-composition \cite{bodlaender2014kernel}]
	\label{def:or-cross-composition}
	Given an \NP-hard problem~$Q$, a parameterized problem~$L$, and a polynomial equivalence relation~$\mathcal R$ on the instances of~$Q$,
	an OR-cross-composition of~$Q$ into~$L$ (with respect to~$\mathcal R$) is an algorithm that takes~$q$ instances~$I_1, I_2, \dots, I_q$ of~$Q$ belonging to the same equivalence class of~$\mathcal{R}$ and constructs in time polynomial in~$\sum_{i=1}^q \abs{I_i}$ an instance~$(I, k)$ of~$L$ such that
	\begin{inparaenum}[(i)]
		\item $k$ is polynomially upper-bounded by $\max_{i \in \oneto{q}} \abs{I_i} + \log(q)$, and
		\item $(I, k)$ is a \yes-instance of~$L$ if and only if there exists an~$i \in \oneto{q}$ such that~$I_i$ is a \yes-instance of~$Q$.
	\end{inparaenum}
\end{definition}
If a parameterized problem admits an OR-cross-composition, then it does not admit a polynomial kernel unless \NPincoNPslashpoly~\cite{bodlaender2014kernel}.
A \emph{polynomial parameter transformation} from a parameterized problems $L$ to a parameterized problem~$L'$ is a parameterized reduction from $L$ to $L'$ that, given an instance~$(I, k)$ of $L$, runs in~$(\abs{I}+k)^{\bigO(1)}$ time and returns an instance~$(I', k')$ of~$L'$ such that~$k'$ is bounded from above by a polynomial in $k$.
If $L$ does not admit a polynomial kernel unless \NPincoNPslashpoly, and there is a polynomial parameter transformation from $L$ to~$L'$, then $L'$ does not admit a polynomial kernel unless \NPincoNPslashpoly~\cite{BTY11}.

\paragraph{Graph parameters and classes.}
We give an overview over the different graph parameters and graph classes used throughout the paper.
We start with the graph parameters ordered from left to right in \cref{parameterHierarchy}.
To this end, let~$G=(V,E)$ be a graph.
The \emph{minimum clique cover} of~$G$ is the minimum number of cliques needed to partition~$V$.
The \emph{maximum diameter of components} is the maximum distance~$\dist(u,v)$ between any two vertices~$u$ and~$v$ in the same connected component of~$G$.
The \emph{average distance} is analogously the average distance~$\dist(u,v)$ between two vertices~$u$ and~$v$ in the same connected component of~$G$.
The \emph{distance to~$\Pi$} for some graph class~$\Pi$ is the minimum number of vertices needed to be removed from~$G$ such that it becomes a graph in~$\Pi$.
The \emph{feedback vertex number} is the distance to forests.
The \emph{treedepth} of an unconnected graph is the maximum treedepth of any of its connected components.
The treedepth of a connected graph~$G = (V, E)$ is defined as follows.
Let~$T$ be a rooted tree with vertex set~$V$ such that if~$\{x, y\} \in E$, then~$x$ is either an ancestor or a descendant of~$y$ in~$T$.
We say that~$G$ is embedded in~$T$.
The depth of~$T$ is the number of vertices in a longest path in~$T$ from the root to a leaf.
The treedepth of~$G$ is the minimum~$t$ such that there is a rooted tree of depth~$t$ in which~$G$ is embedded.
The \emph{maximum leaf number} of~$G$ is the maximum number of leaves among all spanning trees of~$G$.
The \emph{feedback edge number} is the minimum size of a set~$F$ of edges such that removing~$F$ from~$G$ results in a forest.
The \emph{genus} of~$G$ is the minimal integer~$k$ such that~$G$ can be drawn without crossing itself on a sphere with~$k$ handles, that is, on a surface of genus~$k$.
Given an injective function~$f \colon V \to \N$, the bandwidth cost of~$f$ is defined as~$\max_{\{u,v\}\in E}|f(u) - f(v)|$.
The \emph{bandwidth} is the minimum bandwidth cost over all possible injective functions from~$V$ to~$\N$.
We say that two vertices~$u, v$ have the same \emph{type} if~$N(v)\setminus \{u\} = N(u) \setminus\{v\}$.
The \emph{neighborhood diversity} is the minimum number of sets into which~$V$ can be partitioned such that every two vertices in the same partition have the same type.

In this work we consider the following graph classes.
A \emph{cograph} is a graph in which each component has diameter at most two.
Equivalently, a cograph is a graph without induced paths of length three.
An \emph{interval graph} is a graph where each vertex can be represented by an interval of real numbers such that two vertices are adjacent if and only if their respective intervals overlap.
The vertex set of a \emph{bipartite graph} can be partitioned into two independent sets.
A set of \emph{disjoint paths} is a graph in which each connected component is a path.
An \emph{outerplanar graph} is a graph that has a planar drawing (a drawing  in the plane without edge crossings) for which all vertices belong to the outer face of the drawing. 
Finally, a graph is a \emph{cactus graph} if every edge is part of at most one cycle.

\section{Known Results and Basic Observations}
\label{sec:obs}
In this section, we adapt some previously known results and show some simple classification results.
In previous work, we have shown that a generalization of \fsp{} is fixed-parameter tractable with respect to the parameter~$\ell$~\cite[Theorem 4]{BKN22}.
In that generalization, we are not looking for a shortest~\stpath{} but for a path of length at most some given~$\ell$.
Additionally, the difference between the most and least frequent colors can be up to some constant~$\delta$.
By setting~$\ell \defeq \dist(s,t)$ and~$\delta \defeq 1$, we get that \fsp{} is fixed-parameter tractable by the distance between~$s$ and~$t$.
Now observe that any vertex~$v$ with~$\dist(s,v) + \dist(v,t) > \dist(s,t)$ cannot be part of a solution and thus can be removed from the input.
Doing so exhaustively results in a graph of diameter at most~$\dist(s,t)$.
To see this, observe that~$\dist(u,v) \leq \min(\dist(u,s)+\dist(s,v), \dist(u,t)+\dist(t,v))$ for each pair~$u,v$~of vertices.
Since the minimum is always at most the average, we get $$\dist(u,v) \leq \frac{\dist(u,s)+\dist(s,v) + \dist(u,t)+\dist(t,v)}{2} =\dist(s,t).$$
\begin{observation}
	\label{obs:l}
	\fsp{} parameterized by maximum diameter~$d$ of components can be solved in~$(n + m) \log(n) \cdot (4e^2)^{d+2} \cdot d^{\bigO(log(d))}$ time.
\end{observation}

Moreover, we have shown that \fsp{} is \Wone-hard with respect to the number~$c$ of colors \cite[Theorem 3]{BKN22}.
The reduction is from \textsc{Multicolored Clique} parameterized by solution size~$k$ and the distance~$d$ to disjoint paths in the resulting instance is linear in~$k$.
This yields the following.
\begin{observation}
	\label{obs:ddp}
	\fsp{} parameterized by distance to disjoint paths is \Wone-hard.
\end{observation}

Next, we investigate the two parameters feedback vertex number and feedback edge number.
Since a set of disjoint paths is a forest, \cref{obs:ddp} also likely excludes the possibility for fixed-parameter tractability for the feedback vertex number.
We can, however, show containment in~XP.

\begin{observation}
	\label{obs:fvn}
	\fsp{} parameterized by feedback vertex number~$k$ is solvable in~$\bigO(3.460^k \cdot n^{2k+2})$ time.
\end{observation}

\begin{proof}
	We compute a feedback vertex set~$F$ of size~$k$ in~$\bigO(3.460^k \cdot k \cdot n)$~time~\cite{Iwata2021}.
	Next, we guess\footnote{Whenever we pretend to guess something, we iterate over all possibilities and consider the correct iteration.} the vertices~$F' \subseteq F$ belonging to a solution.
	We then order these vertices with respect to their distance to terminal~$s$.
	Note that this ordering is strict or otherwise the guess is wrong.
	Hence, it only remains to find the correct vertices outside of~$F$ to connect each pair of consecutive vertices in~$F'$.
	To this end, we guess the first vertex after and the last vertex before each vertex in~$F'$ in a solution path.
	Observe that~$G-F$ is a forest and thus there is a unique path in~$G-F$ between each pair of vertices in~$V \setminus F$ (or no path if the two vertices belong to different connected components).
	Thus, any guess can be extended to a unique~$s$-$t$-path in linear time.
	We return \yes{} if and only if one of these paths is a \fair{} shortest~$s$-$t$-path.
	Note that if a \fair{} shortest~$s$-$t$-path~$P$ exists, then we find it as it corresponds to the guess that considers all vertices from~$P$ that are in~$F$ or adjacent to a vertex in~$F$.
	Finally, since the number of possibilities to guess from is upper-bounded by~$2^{|F|} \cdot n^{2|F|} = 2^k \cdot n^{2k}$, the overall running time is in~$\bigO(3.460^k \cdot k \cdot n + 2^k \cdot n^{2k} \cdot (n+m)) \subseteq \bigO(3.460^k \cdot n^{2k+2})$.
\end{proof}

For the closely related feedback edge number, we can show a stronger result. 

\begin{observation}
	\label{obs:fen}
	\fsp{} parameterized by feedback edge number~$k$ is solvable in~$\bigO(2^k \cdot (n+m))$ time.
\end{observation}

\begin{proof}
	We assume without loss of generality that the input instance is connected.
	First, we compute a minimum feedback edge set~$F$ in linear time by computing a minimum spanning tree and taking all remaining edges.
	Second, we guess the set~$F' \subseteq F$ of feedback edges contained in a solution path.
	Let~$V' \defeq \bigcup_{e \in F'} e$ be the set of endpoints of edges in~$F'$.
	We sort the vertices in~$V'$ by their distance from~$s$ in linear time using bucket sort.
	Since~$G' \defeq (V,E\setminus F)$ is a tree, it holds for each pair~$u,v$ of consecutive vertices in the order that if they are not connected by an edge in~$F'$, then there is a unique~$u$-$v$-path in~$G'$.
	Hence, any guess can be extended to a unique~$s$-$t$-path in linear time.
	If any of these paths is a \fair{} shortest~$s$-$t$-path, then we return \yes, otherwise we return \no.
	Note that if a \fair{} shortest~$s$-$t$-path~$P$ exists, then we find it as it is corresponds to the guess that considers all edges from~$P$ that are in~$F$.
	Finally, since the number of possibilities to guess from is~$2^{|F|} = 2^k$, the overall running time is~$\bigO(2^k \cdot (n+m))$.
\end{proof}

We end this section with the parameter average distance.
We show that adding a large clique to the input yields para-NP-hardness for this parameter.

\begin{observation}
	\label{obs:dist}
	\fsp{} is \NP-hard even for instances where the average distance between vertices is at most~$3$.
\end{observation}

\begin{proof}
	We provide a reduction from \fsp.
	We add a clique of~$n^2$ vertices and make each of its vertices adjacent to the terminal~$s$.
	First, note that none of these vertices can be part of any shortest~$s$-$t$-path and thus it is an equivalent instance.
	Hence, it only remains to analyze the average distance between vertices.
	To this end, we denote by~$V$ the set of the~$n$ original vertices and by~$K$ the set of the~$n^2$ new vertices.
	The maximum distance between two vertices in~$V$ is at most~$n$, the maximum distance between two vertices in~$K$ is one, and the maximum distance between a vertex in~$V$ and a vertex in~$K$ is~$n$.
	Thus, the average distance is at most
	\[ \frac{\binom{n}{2} \cdot n + \binom{n^2}{2} \cdot 1 + n \cdot n^2 \cdot n}{\binom{n^2 + n}{2}} \leq \frac{n^4 + n^4 + n^4}{n^4} = 3. \qedhere\]
\end{proof}

\section{Kernelization}
\label{sec:kernel}
In this section, we present all our results related to polynomial kernels, that is, we show a polynomial kernel for distance to cographs (and more generally for distance to~$P_h$-free graphs for any constant~$h$) and we exclude polynomial kernels for minimum clique cover and for treedepth.
We also show that a slight generalization of \fsp{} regarding a generalized fairness parameter~$\delta$ does not allow for a polynomial kernel for the parameter maximum leaf number plus~$\delta$.
We leave open whether \fsp{} parameterized by maximum leaf number admits a polynomial kernel.
Let us start by providing a polynomial kernel for \fsp{} parameterized by the neighborhood diversity.
\begin{proposition}
	\label{prop:kernd}
	\fsp{} admits a kernel containing at most~$k^2$ vertices, where~$k$ is the neighborhood diversity.
	The kernel is computable in~$\bigO(n^2 k)$ time.
\end{proposition}

\begin{proof}
	Let~$I = (G=(V,E),\chi,s,t)$ be an instance of \fsp{} and let~$\ell$ be the distance between~$s$ and~$t$ in~$G$.
	We start off by computing a neighborhood partition of minimum width~$k$, which can be achieved in $\bigO(n^2 k)$ time \cite{lampis2012meta}.
	The crucial observation towards obtaining the kernel now is that a shortest path does not visit two vertices~$u, v$ of the same type:
	If a path~$P$ first visits~$u$ and later~$v$, then the predecessor of~$u$ is also adjacent to~$v$; thus the path can be short-cut.
	Hence, $\ell \le k$, and it suffices to keep one vertex of each color for each type.
	If the number~$c$ of colors is at most~$k$, then this yields a kernel at most~$k^2$ vertices.
	If~$c > k$, then we derive from~$k \ge \ell$ that the sought shortest path will contain at most one vertex of each color.
	In this case, it suffices to keep~$k$ vertices of different colors for each type:
	No matter how the remaining~$\ell-1 \le k-1$ vertices of other types were chosen, there is one vertex of our type which has a different color.
	Hence, if there is a type that contains at least~$k$ vertices of different colors, we can remove the remaining vertices of that type.
	This again yields a kernel with at most~$k^2$ vertices.
	The running time is dominated by determining the neighborhood diversity.
\end{proof}

We next provide a polynomial kernel for \fsp{} parameterized by the distance to~$P_h$-free graphs for any constant~$h$.

\begin{theorem}
	\label{fp:thm:p-h-free}
	Let~$h \in \NN$.
	Then \fsp{} parameterized by the vertex deletion distance~$k$ to~$P_h$-free graphs admits a kernel with at most~$\bigO(h^{h+2}(k+1)^{h+1})$ vertices.
	The kernel can be computed in~$\bigO(n^{\omega h} + n^4)$ time.
\end{theorem}

\begin{proof}
	Let~$I=(G=(V,E),\chi,s,t)$ be an instance of \fsp{}.
	Let~$\ell$ be the number of vertices in a shortest $s$-$t$-path, that is one more than the distance between~$s$ and~$t$ in~$G$.
	Let~$K' \subseteq V$ be a set of~$k$ vertices such that~$K'$ is~$P_h$-free.
	Note that each connected component of~$G - K'$ has diameter at most~$h-2$, thus any shortest~$s$-$t$-path can contain at most~$h-1$ vertices from any such connected component.
	Thus, at least each~$h$\textsuperscript{th} vertex on a shortest~$s$-$t$-path has to be contained~$K$.
	Hence, we have~$\ell \leq h(k+1)$.
	Next, we compute a set~$K \subseteq V(G)$ with~$|K| \le hk$ such that~$G-K$ is~$P_h$-free in~$\bigO(k \cdot n^h)$ time.
	We can find an induced~$P_h$ in~$\bigO(n^h \cdot h^2)$ time and any vertex deletion set to~$P_h$-free graphs contains at least one of the~$h$ vertices.
	Taking all $h$ vertices into the set~$K$ is therefore an~$h$-approximation.
	For convenience, we will also add~$s$ and~$t$ to~$K$.
	Next, we compute all shortest paths between any pair~$u, v \in K$ which do not contain any other vertices from~$k$.
	Each such path~$P$ visits the same number~$d_{uv} \defeq \dist_G(u, v) - 1 \le h-1$ of vertices, all of which are in some connected component of~$G-K$.
	We replace the graph~$G-K$ by the union of all paths computed and we remove duplicate paths, that is, if there are two~$u$-$v$-paths which use the same multiset of colors, then we remove one of them.
	Let~$\SSS_{uv}$ be the set of remaining~$u$-$v$-paths.
	
	We now make a case distinction on the number~$c$ of colors in the input graph.
	If~$c \leq \ell$, then note that~$|\SSS_{uv}| \leq c^{h-1} \leq (h(k+1))^{h-1} = h^{h-1} (k+1)^{h-1}$.
	Since each path in~$\SSS_{uv}$ contains~$d_{uv} < h$ vertices, the entire remaining graph contains at most~${|K| + |K|^2 h^h (k+1)^{h-1} \in \bigO(h^{h+2}(k+1)^{h+1})}$ vertices.
	
	If~$c > \ell$, then each color can appear at most once in any solution as by the pigeonhole principle, at least one color does not appear at all.
	Hence, for each pair~$u,v \in K$, we remove any path from~$\SSS_{uv}$ that contains the same color twice.
	Note that this implies~$|\SSS_{uv}| \leq n^h$.
	Finally, we use the technique of representative families by Fomin et al.~\cite{fomin2016representative} to shrink the size of each set~$\SSS_{uv}$ even further.
	Let~$M = (V, \III)$ be a matroid, where the set~$V$ of vertices in~$G$ is the ground set and a set~$X \subset V$ is independent in~$M$ if and only if each color appears at most once in~$X$ and~$|X| \leq \ell$.
	More formally, the set~$\III$ is defined as follows.
	\begin{equation*}
		\III = \{X \subseteq V \mid \abs{X \cap \col^i} \le 1 \text{ for } i \in \colors \text{ and } |X| \le \ell\}.
	\end{equation*}
	Recall that~$\chi^i$ is the set of vertices with color~$i$.
	Note that~$M$ is a partition matroid of rank~$\ell$;
	therefore, we can compute a linear representation over a field of size~$\bigO(n)$ in~$\bigO(n^4)$ time \cite{Marx09}.
	We can then use \cref{thm:repr} (with a uniform weight function) to compute, for each set~$\SSS_{uv}$, a~$(\ell-d_{uv})$-representative family~$\hat{\SSS}_{uv}$ of size at most~$\binom{\ell}{d_{uv}} \leq \ell^{h-1}$ in~$\bigO(n^h \ell^{h(\omega-1)}) \subseteq \bigO(n^{\omega h})$~time, where~$\omega$ is the fast matrix multiplication constant.
	We return the kernel that contains all vertices in~$K$ and all paths in~$\hat{S}_{uv}$ for each pair~$u,v\in K$ (which might be a single edge if~$\{u,v\} \in E$).
	Note that the number of vertices in the kernel is at most
	\[
		|K| + |K|^2 \ell^{h-1} h \leq kh+2 + (kh+2)^2 (h(k+1))^{h-1} h \in \bigO(h^{h+2}(k+1)^{h+1}).
	\]
	
	It remains to show that the returned instance is equivalent to the input instance.
	In the first direction, note that any path in the kernel corresponds to a path in the original graph that collects the same multiset of colors.
	Moreover, such a path is a shortest~$s$-$t$-path in the kernel if and only if it is a shortest~$s$-$t$-path in the original instance.
	For the reverse direction, let~$P$ be some balance-fair shortest~$s$-$t$-path in~$G$.
	If~$c \leq \ell$, then~$P$ directly corresponds to a path in the kernel that uses the same number of vertices and the same multiset of colors.
	If~$c > \ell$, then let~$L=\{s=v_0,v_1,\ldots,v_p,t=v_{p+1}\} \subseteq K$ be the set of vertices in~$K$ that appear in~$P$.
	We assume that the vertices are ordered in such a way that~$v_i$ appears earlier than~$v_j$ in~$P$ whenever~$i < j$.
	We iteratively replace subpaths of~$P$ by paths in the kernel.
	To this end, let~$P_0=P$ and for each~$1 \leq i \leq p+1$, let~$P_i$ be constructed as follows.
	Let~$Q_i$ be the set of colors in~$P_{i-1}$ that do not include the set of vertices in the path between~$v_{i-1}$ and~$v_i$ ($Q_i$ contains the colors of~$v_{i-1}$ and~$v_i$).
	Let~$S_i$ be all remaining colors in~$P_{i-1}$.
	By definition, there is a set~$S \in \hat{\SSS}_{v_{i-1}v_i}$ such that~$S \cup Q_i \in \III$, that is, there is a path in~$\SSS_{v_{i-1}v_i}$ whose colors are disjoint from~$Q_i$.
	We replace the subpath of~$P_{i-1}$ between~$v_{i-1}$ and~$v_i$ by the subpath with colors in~$S$ to get path~$P_i$.
	Note that~$P_{p+1}$ is then a shortest~$s$-$t$-path in the kernel that uses each color at most once.
	This concludes the proof.
\end{proof}

Since cographs are exactly the~$P_4$-free graphs, this implies a polynomial kernel for distance to cographs.

\begin{corollary}
	\label{prop:kercograph}
	\fsp{} admits a kernel of size~$\bigO(k^5)$, where~$k$~is the distance to cographs.
	The kernel can be computed in~$\bigO(n^{10})$ time.
\end{corollary}

The maximum diameter of components is a parameter directly under the distance to cographs in the parameter hierarchy (see \cref{parameterHierarchy}).
As shown in \cref{obs:l}, \fsp{} parameterized by maximum diameter of components is fixed-parameter tractable.
We next show that it presumably does \emph{not} allow for a polynomial kernel by excluding polynomial kernels for \fsp{} parameterized by two parameters upper-bounding maximum diameter of components---treedepth and minimum clique cover.
We first show that \fsp{} parameterized by treedepth admits an OR-cross-composition.

\begin{proposition}
	\label{prop:nokerntd}
	\fsp{} parameterized by treedepth does not admit a polynomial kernel unless~\NPincoNPslashpoly.
\end{proposition}
\begin{proof}
	We prove the statement by showing that \fsp{} OR-cross-composes into \fsp{} parameterized by treedepth.
	We first define the polynomial equivalence relation $\mathcal R$.
	Two instances~${(G=(V,E), \chi \colon V \to \oneto{c}, s, t)}$ and~${(G'=(V',E'), \chi' \colon V' \to \oneto{c'}, s', t')}$ of \fsp{} are~$\mathcal R$-equivalent if and only if~$c = c'$ and ${\dist_G(s, t) = \dist_{G'}(s', t')}$.
	Clearly, the relation fulfills condition (i) of an equivalence relation (see \cref{sec:prelim}).
	Since both integers considered for equivalence are upper-bounded in the size of their respective instance, condition (ii) holds as well.
	Consider~$q$~instances~$I_1,I_2,\ldots,I_q$ of \fsp, where~$I_i \defeq (G_i, \chi_i, s_i, t_i)$ for all~$i \in \oneto{q}$ and all instances are from the same equivalence class of~$\mathcal R$.
	Let~$\ell \defeq \dist_{G_i}(s_i, t_i)+1$ be the number of vertices in each shortest path in each graph~$G_i$, and let~$c$ be the number of colors used in each instance~$I_i$.
	Further, let~$x \defeq \lfloor \nicefrac{\ell}{c} \rfloor$ and~$x' \defeq \lceil \nicefrac{\ell}{c}\rceil$ be the minimum and maximum number of appearances of a color in a \fair{} path in $I_i$.
	We OR-cross-compose into one instance~$I \defeq (G, \chi, s, t)$ of \fsp{} parameterized by the treedepth of~$G$.

	The graph $G$ contains for each~$i \in \oneto{p}$ the graph~$G_i$ as an induced subgraph.
	We make the terminal~$s$ adjacent to each~$s_i$.
	Next, we introduce a path on~$\max (x'-1, 1)$ vertices, one endpoint of which is adjacent to each~$t_i$, and the other endpoint of which is~$t$.
	Our coloring~$\chi$ adopts the colorings of~$\chi_i$ for each of the vertices in~$G_i$.
	Terminal~$s$ is assigned color~$c+1$.
	If~$x' \le 1$, then the path consists only of $t$ and we set~$\chi(t) = c+2$.
	Otherwise, we assign color~$c+1$ to each vertex of the path.

	The construction clearly runs in time polynomial in the sum of the input instance sizes.
	The treedepth of $G$ is at most $$x'+ 1 + \max_{i \in \oneto{q}} \text{td}(G_i) \le 2 \max_{i \in \oneto{q}} \abs{I_i},$$ wherein~$\text{td}(G_i)$ is the treedepth of~$G_i$.
	Hence, our construction fulfills property~(i) of \cref{def:or-cross-composition}.
	We next show that it also fulfills property~(ii), that is, $I$ is a \yes-instance if and only if there exists an~$i \in \oneto{p}$ such that~$I_i$ is a \yes-instance.
	To this end, observe that $\dist_{G}(s, t) = \ell + \max(x', 2)$, and that for each~$i \in \oneto{q}$ and each shortest \abpath{s_i}{t_i} $P_i$, the (unique) \stpath{} that contains~$P_i$ is a shortest path in $G$.
	
	Suppose first that one of the input instances is a \yes-instance, that is, one of the graphs~$G_i$ contains a shortest \fair{}~\mbox{\abpath{s_i}{t_i}}~$P_i$.
	Then, each color~$j \in \oneto{c}$ appears $x_j$ times in $P_i$, wherein~${x \le  x_j \le x'}$.
	We show that the \stpath{} $P$ which contains $P_i$ is \fair{} as well.
	If~$x' \le 1$, then~$\ell \leq c$, each color~$j \in \oneto{c}$ is visited at most once, and the only additional vertices are~$s$ and~$t$, which both have unique colors.
	If~$x' > 1$, then there are exactly~$x'$ vertices in~$P$ that are not in~$P_i$ and all of those have color~$c+1$.
	As this color does not appear in~$P_i$, the path~$P$ is \fair.
	Hence, $I$ is a \yes-instance.

	Suppose next that~$I$ is a \yes-instance, that is, there exists a shortest \fair{} \stpath{} $P$.
	Such a path $P$ visits vertices from exactly one of the graphs~$G_i$, let~$P_i$ be the corresponding subpath from~$s_i$ to~$t_i$.
	If $x' \le 1$, then $P$ contains $\ell+2$ vertices and $\chi$ assigns by construction~$c+2$ colors to the vertices in~$G$.
	Hence, $P$ must contain at most one vertex of each color.
	As~$\chi(s) = c+1$ and~$\chi(t) = c+2$, the subpath~$P_i$ must contain each color in~$\oneto{c}$ at most once and it is therefore a shortest \fair{} \abpath{s_i}{t_i} and~$I_i$ is a \yes-instance.
	If~$x' > 1$, then $P$ contains~$\ell+x'$ vertices and~$\chi$ assigns~$c+1$ colors to the vertices in~$G$.
	Since (i) there are exactly $x'$ vertices in~$P$ that are not in~$P_i$ and (ii) $P_i$ consists of exactly $\ell$ vertices whose colors are in~$\oneto{c}$, it must contain each of these colors at least~$x$ and at most~$x'$ times.
	Thus, $P_i$ is a shortest \fair{} \abpath{s_i}{t_i} and~$I_i$ is a \yes-instance.
\end{proof}

We only need a slight modification of the above OR-cross-composition to show that there is presumably no polynomial kernel for \fsp{} when parameterized by minimum clique cover.

\begin{proposition}
	\label{prop:nokernmcc}
	\fsp{} parameterized by minimum clique cover does not admit a polynomial kernel unless~\NPincoNPslashpoly.
\end{proposition}
\begin{proof}
	We extend the construction shown in the proof of \cref{prop:nokerntd} to show that \fsp{} OR-cross-composes into \fsp{} parameterized by minimum clique cover.
	We use the same polynomial equivalence relation.
	Given $q$ instances~$I_1,I_2,\ldots,I_q$ with~$I_i = (G_i, \chi_i, s_i, t_i)$ for all~$i \in \oneto{q}$ from the same equivalence class in~$\mathcal R$, let~$I = (G=(V,E), \chi, s, t)$ be the instance constructed in the proof of \cref{prop:nokerntd}.
	We extend the construction as follows.
	Let~$\ell = \dist_G(s, t)$.
	For each~$d \in \oneto{\ell}$, let~$V_d \defeq \{v \in V \mid \dist_G(s, v) = d\}$ be the set of vertices with distance exactly~$d$ from~$s$.
	For each~$V_d$, we add an edge between each pair of vertices in~$V_d$.
	Let~$E_d$ be the set of these edges for each~$d \in \oneto{\abs{V}}$.

	Clearly, the composition can still be computed in time polynomial in the sum of the input instance sizes.
	As for property~(i) of \cref{def:or-cross-composition}, note that~$G$ now contains a minimum clique cover of size~$\ell$, consisting of the cliques on~$V_d$ for each~$d \in \oneto{\ell}$ (and~$s$ being part of the clique on~$V_1$).
	This is trivially linearly upper-bounded in~$\max_{i \in \oneto{q}} \abs{I_i}$.
	Lastly, note that each \stpaths{} or \abpaths{s_i}{t_i} for some~$i \in \oneto{q}$ containing an edge from a set~$E_d$ is not a shortest path.
	Hence, our extension does not affect any of the shortest paths of interest, and our construction fulfills property (ii) of \cref{def:or-cross-composition}.
\end{proof}

Concluding this section, we investigate the parameter maximum leaf number.
We show that \fsp{} does not admit a polynomial kernel in the maximum leaf number by providing a polynomial parameter transformation from \prob{Exact Hitting Set}.
\prob{Exact Hitting Set} admits no problem kernel of size polynomial in~$\abs{U}$ unless \NPincoNPslashpoly~\cite{DomLS14}.\footnote{We mention that the authors only exclude a polynomial kernel for \textsc{Hitting Set} parameterized by~$|U|$, but all of the arguments work exactly the same for \textsc{Exact Hitting Set}.}
It is defined as follows.

\problemdef{\prob{Exact Hitting Set}}
{A universe~$U$ and a family~$\calF$ of subsets of~$U$.}
{Is there a subset~$X \subseteq U$ such that $|S \cap X| = 1$ for each~${S \in \calF}$?}

\begin{proposition}
	\label{prop:pkfen}
	\fsp{} parameterized by the maximum leaf number does not admit a polynomial kernel unless~\NPincoNPslashpoly.
\end{proposition}

\begin{proof}
	Let~$(U,\calF)$ be an instance of \prob{Exact Hitting Set}.
	We provide a polynomial parameter transformation from \prob{Exact Hitting Set} with parameter~$d \defeq \abs{U}$.
	We construct an instance~$I = (G=(V,E),\chi,s,t)$ of \fsp{} as follows.
	We will use~$\abs{\calF} + 4$ colors:
	a color~$p_S$ for each set~$S \in \calF$, and three additional colors~$r_\ell$, $r_h$, and~$q$.  
	The graph~$G$ consists of a path~$G_P$, a gadget~$G_x$ for each element~$x \in U$, and the vertex~$t$ with color~$r_h$.
	For an element~$x \in U$, let~$s_x$ be the number of sets in~$\calF$ containing~$x$ and let~$\calF_x = \{S_1^x,S_2^x,\ldots S_{s_x}^x\}$ be the set of sets in~$\calF$ that contain~$x$.
	The gadget~$G_x$ consists of two disjoint paths~$Y_x = (y_1,y_2,\ldots,y_{s_x})$ and~$N_x=(n_1,n_2,\ldots,n_{s_x})$,
	wherein each vertex~$y_i$ with~$i \in \oneto{s_x}$ receives color~$p_{S_i^x}$ and each vertex~$n_j$ with~$j \in \oneto{s_x}$ receives color~$q$.
	Let~$a = (\sum_{S \in \calF}|S|)-d$.
	The path~$G_P$ starts with $a$~vertices of color~$r_\ell$ and~$a$ vertices of color~$r_h$, the first of which is the terminal~$s$.
	It then proceeds with~$a-1$ vertices of each color~$p_S$ for each~$S \in \calF$.
	We then arrange all gadgets~$G_x$ in a line and connect each last vertex of the two paths of a gadget with each first vertex of the two paths of the next gadget.
	Moreover, we connect the last vertex of~$G_P$ with the first vertices in the two paths of the first gadget and the last vertices in the two paths of the last gadget with~$t$.
	
	Note that for each induced path in the resulting graph, there can be at most two leaves in any spanning tree.
	Since all edges can be partitioned into~$2d+1$~paths (two in the gadget for each element~$x\in U$ and one is~$G_P$), the maximum leaf number of the resulting instance is at most~$4d+2$ (which is polynomial in~$d$).
	Moreover, since the reduction can clearly be computed in polynomial time, it only remains to show that the constructed instance of \fsp{} is a \yes-instance if and only if the original instance of \prob{Exact Hitting Set} is a \yes-instance.
	To this end, observe that each shortest~\stpath{} contains all vertices in~$G_P$, the vertex~$t$, and either all vertices of~$Y_x$ or all vertices of~$N_x$ for each~$x \in U$.
	Hence, each \stpath{} contains exactly~$a$ vertices of color~$r_\ell$ and~$a+1$ vertices of color~$r_h$.
	Thus, every solution path must contain between $a$~and~$a+1$ vertices of each color.
	
	Assume first that there is an exact hitting set~$X \subseteq U$.
	Consider the~\stpath~$P$ in~$G$ that contains all vertices in~$G_P$ and, for each gadget~$G_x$, the vertices in~$Y_x$ if~$x \in X$ and the vertices in~$N_x$ if~$x \notin X$.
	Since~$X$ is an exact hitting set, the path contains exactly one vertex of color~$p_S$ for each set~$S \in \calF$.
	Since the path~$G_P$ contains~$a-1$ vertices of each such color, each such color appears exactly~$a$ times in~$P$.
	Finally, note that the number of vertices with color~$q$ in~$P$ is also~$a$ as~$P$ contains all~$\sum_{S \in \calF}|S| = a + d$ vertices in all paths~$N_x$ except for the exactly~$d$ vertices where it contains vertices in the~$Y_x$-paths.
	Thus, the path~$P$ verifies that $I$ is a \yes-instance.
	
	Now assume that~$I$ is a \yes-instance.
	Then, there is a path~$P$ containing all vertices in~$G_P$, the vertices of~$Y_x$ or the vertices of~$N_x$ for each~${x\in U}$, and the vertex~$t$, such that each color appears between~$a$ and~$a+1$ times in~$P$.
	Consider the set~$X \subseteq U$ containing an element~${x \in U}$ if and only if~$P$ contains the vertices of~$Y_x$.
	Since~$P$ contains each color at least~$a$ times, it contains at least one vertex of each color~$p_S$ in some subpath~$Y_x$, that is,~$X$ is a hitting set.
	Moreover, assume that it holds for some set~$S \in \calF$ that~$|S \cap X| > 1$.
	Then,~$P$ contains at most~$\sum_{S \in \calF} |S| - d - 1 = a-1$ vertices of color~$q$, a contradiction.
	Thus, $X$ is an exact cover.
	This concludes the proof.
\end{proof}

\section{Para-NP-Hardness}
\label{sec:paraNP}
We complete our tetrachotomy by showing para-NP-hardness for several parameters.
We start by showing that \fsp{} is \NP-hard on interval graphs by providing a reduction from \prob{Vertex Cover}.
Afterwards, we conclude with a reduction from \textsc{Exact Cover} showing \NP-hardness for bipartite outerplanar graphs with genus zero and bandwidth two.
\textsc{Vertex Cover} is \NP-hard and defined as follows~\cite{Kar72}.

\problemdef{\prob{Vertex Cover}}
{A graph~$G = (V, E)$ and an integer~$k$.}
{Is there a set~$K \subseteq V$ with~$\abs{K}\le k$ such that~$G - K$ is edgeless?}

\begin{proposition}
	\label{prop:interval}
	\fsp{} is \NP-hard even on interval graphs.
\end{proposition}

\begin{proof}
	Let~$(H=(V,E),k)$ be an instance of \prob{Vertex Cover}.
	Moreover, let~$V = \{v_1, v_2 \dots, v_n\}$ and~$E = \{e_1, e_2 \dots, e_m\}$.
	We construct an interval graph~$G=(V',E')$ with two vertices~$s,t$ and a coloring function~${\chi \colon V' \rightarrow \oneto{n+3}}$ such that there is a \fair{} shortest~$s$-$t$-path in~$G$ if and only if there is a vertex cover of size~$k$ in~$H$.
	We will use one color~$p_v$ for each vertex~$v \in V$ and three additional colors~$q_1$, $q_2$, and~$r$.
	The graph~$G$ consists of a path~$G_P$ and three types of gadgets: a vertex gadget~$G_v$ for each vertex~$v \in V$, an edge gadget~$G_e$ for each edge~$e \in E$, and a filler gadget~$G_f$.
	We next describe the different parts in more detail.
	
	The vertex gadget~$G_v$ consists of two parallel paths ${U_v = (u_1,u_2,\ldots,u_{m+1})}$ and~$L_v=(\ell_1,\ell_2,\ldots,\ell_{m+1})$ of length~$m$.
	Additionally, there is an edge between~$\ell_i$ and~$u_i$ for each~$i \in \oneto{m+1}$ and an edge between~$\ell_i$ and~$u_{i+1}$ for each~$i \in \oneto{m}$.
	Lastly, there are two vertices~$s_v$ and~$t_v$ of color~$r$ where~$s$ is adjacent to~$u_1$ and~$\ell_1$ and~$t$ is adjacent to~$u_{m+1}$ and~$\ell_{m+1}$.
	Each vertex in~$U_v$ gets color~$p_v$, vertex~$\ell_{m+1}$ gets color~$q_2$, and all other vertices in~$L_v$ get color~$q_1$.
	See \cref{fig:intervalvertex} for an example.
	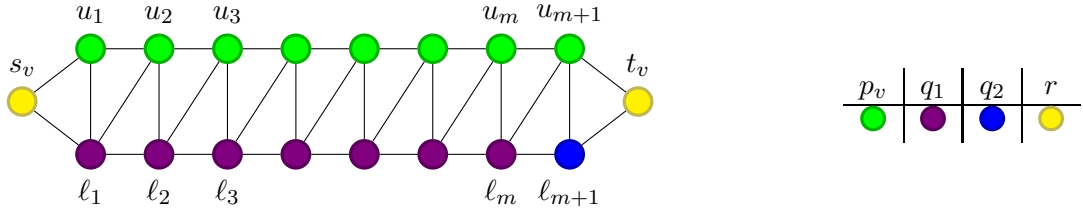
\begin{figure}
	\centering
	\begin{minipage}[c]{0.72\textwidth}
	\begin{tikzpicture}[xscale=.9,yscale=.7]
		\node[colornode=yellow, label=$s_v$] at (-4.5,0) (s) {};
		\node[colornode=yellow, label=$t_v$] at(4.5,0) (t) {};
                                                      ,
		\node[colornode=violet, label=below:$\ell_1$] at (-3.5,-1) (p11) {} edge(s);
		\node[colornode=violet, label=below:$\ell_2$] at (-2.5,-1) (p12) {} edge(p11);
		\node[colornode=violet, label=below:$\ell_3$] at (-1.5,-1) (p13) {} edge(p12);
		\node[colornode=violet] at (-.5,-1) (p14) {} edge(p13);
		\node[colornode=violet] at (.5,-1) (p15) {} edge(p14);
		\node[colornode=violet] at (1.5,-1) (p16) {} edge(p15);
		\node[colornode=violet, label=below:$\ell_m$] at (2.5,-1) (p17) {} edge(p16);
		\node[colornode=blue, label=below:$\ell_{m+1}$] at (3.5,-1) (p18) {} edge(p17) edge(t);
		
		\node[colornode=green, label=$u_1$] at (-3.5,1) (p21) {} edge(s) edge(p11);
		\node[colornode=green, label=$u_2$] at (-2.5,1) (p22) {} edge(p21) edge(p11) edge(p12);
		\node[colornode=green, label=$u_3$] at (-1.5,1) (p23) {} edge(p22) edge(p12) edge(p13);
		\node[colornode=green] at (-.5,1) (p24) {} edge(p23) edge(p13) edge(p14);
		\node[colornode=green] at (.5,1) (p25) {} edge(p24) edge(p14) edge(p15);
		\node[colornode=green] at (1.5,1) (p26) {} edge(p25) edge(p15) edge(p16);
		\node[colornode=green, label=$u_m$] at (2.5,1) (p27) {} edge(p26) edge(p16) edge(p17);
		\node[colornode=green, label=$u_{m+1}$] at (3.5,1) (p28) {} edge(p27) edge(p17) edge(p18) edge(t);
	\end{tikzpicture}%
	\end{minipage}
	\hfill
	\begin{minipage}[c]{0.26\textwidth}
	\begin{tabular}{c|c|c|c}
		$p_v$ & $q_1$ & $q_2$ & $r$\\\hline
		\tikz{\node[colornode=green, inner sep=3pt] at (0,0) (s) {};}&
		\tikz{\node[colornode=violet,  inner sep=3pt] at (0,0) (s) {};}&
		\tikz{\node[colornode=blue,inner sep=3pt] at (0,0) (s) {};}&
		\tikz{\node[colornode=yellow,inner sep=3pt] at (0,0) (s) {};}
	\end{tabular}
	\end{minipage}
	\caption{The vertex gadget~$G_v$ for a vertex~$v$ and a legend providing the names of each color.}
	\label{fig:intervalvertex}
	\end{figure}%
	The gadget for an edge~$e=\{u,v\}$ is a diamond where the two vertices~$s_e$ and~$t_e$ of degree two have color~$r$ and the other two vertices have colors~$p_u$ and~$p_v$, respectively.
	See \cref{fig:intervaledge} for an illustration.
	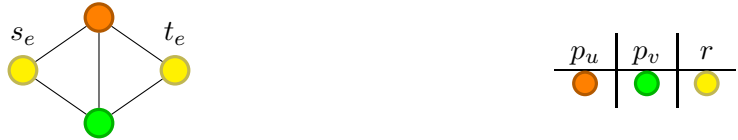
\begin{figure}[t]
	\centering
	\begin{minipage}[c]{0.7\textwidth}\begin{center}
	\begin{tikzpicture}[yscale=.7]
		\node[colornode=yellow, label=$s_e$] at (-1,0) (s) {};
		\node[colornode=yellow, label=$t_e$] at(1,0) (t) {};
		
		\node[colornode=orange] at (0,1) (u) {} edge(s) edge(t);
		\node[colornode=green] at (0,-1) (v) {} edge(s) edge(t) edge (u);
	\end{tikzpicture}%
	\end{center}\end{minipage}
	\hfill
	\begin{minipage}[c]{0.25\textwidth}
	\begin{tabular}{c|c|c}
		$p_u$ & $p_v$ & $r$\\\hline
		\tikz{\node[colornode=orange,inner sep=3pt] at (0,0) (s) {};}&
		\tikz{\node[colornode=green,inner sep=3pt] at (0,0) (s) {};}&
		\tikz{\node[colornode=yellow,inner sep=3pt] at (0,0) (s) {};}
	\end{tabular}
	\end{minipage}
	\caption{The vertex gadget~$G_e$ for an edge~$e = \{u,v\}$ and a legend providing color names.}
	\label{fig:intervaledge}
	\end{figure}
	
	The filler gadget~$G_f$ consists of~$k \cdot (m+1) - m$ levels, where each level contains for each vertex~$v\in V$ a vertex of each color~$p_v$.
	Each level induces a clique and all vertices of two consecutive levels are pairwise adjacent.
	\cref{fig:intervalfiller} illustrates~$G_f$. 
	\begin{figure}[t]
	\centering
	\begin{tikzpicture}[yscale=.7]
		\node[colornode=orange] at (-3,2) (o1) {};
		\node[colornode=green] at (-3,1) (o2) {} edge(o1);
		\node[colornode=cyan] at (-3,0) (o3) {} edge[bend left=22](o1) edge(o2);
		\node[colornode=brown] at (-3,-1) (o4) {} edge[bend right=22](o1) edge[bend left=22](o2) edge(o3);
		\node[colornode=r0] at (-3,-2) (o5) {} edge[bend left=30](o1) edge[bend right=22](o2) edge[bend left=22](o3) edge(o4);

		\node[colornode=orange] at (0,2) (p1) {} edge(o1) edge(o2) edge(o3) edge(o4) edge(o5);
		\node[colornode=green] at (0,1) (p2) {} edge(p1) edge(o1) edge(o2) edge(o3) edge(o4) edge(o5);
		\node[colornode=cyan] at (0,0) (p3) {} edge[bend left=22](p1) edge(p2) edge(o1) edge(o2) edge(o3) edge(o4) edge(o5);
		\node[colornode=brown] at (0,-1) (p4) {} edge[bend right=22](p1) edge[bend left=22](p2) edge(p3) edge(o1) edge(o2) edge(o3) edge(o4) edge(o5);
		\node[colornode=r0] at (0,-2) (p5) {} edge[bend left=30](p1) edge[bend right=22](p2) edge[bend left=22](p3) edge(p4) edge(o1) edge(o2) edge(o3) edge(o4) edge(o5);
		
		\node[colornode=orange] at (3,2) (q1) {} edge(p1) edge(p2) edge(p3) edge(p4) edge(p5);
		\node[colornode=green] at (3,1) (q2) {} edge(q1) edge(p1) edge(p2) edge(p3) edge(p4) edge(p5);
		\node[colornode=cyan] at (3,0) (q3) {} edge[bend left=22](q1) edge(q2) edge(p1) edge(p2) edge(p3) edge(p4) edge(p5);
		\node[colornode=brown] at (3,-1) (q4) {} edge[bend right=22](q1) edge[bend left=22](q2) edge(q3) edge(p1) edge(p2) edge(p3) edge(p4) edge(p5);
		\node[colornode=r0] at (3,-2) (q5) {} edge[bend left=30](q1) edge[bend right=22](q2) edge[bend left=22](q3) edge(q4) edge(p1) edge(p2) edge(p3) edge(p4) edge(p5);
	\end{tikzpicture}
	\caption{An extract from the filler gadget~$G_f$.}
	\label{fig:intervalfiller}
	\end{figure}
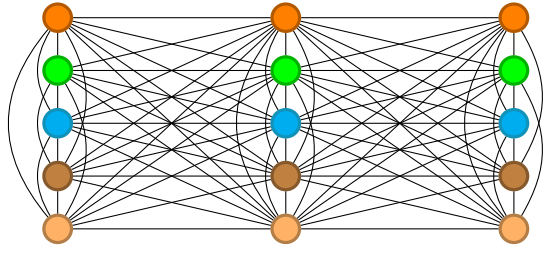
	
	The path~$G_P$ contains~$k m - k$ vertices of colors~$q_2$, $k m - 2(m+n)-1$~vertices of color~$r$, and $k m - (m+1)$ vertices of color~$p_v$ for each vertex~$v\in V$.
	One of the vertices of color~$r$ is~$s$ and this is the first vertex in~$G_P$.
	To finish the construction of~$G$, we add the vertex~$t$ with color~$r$ and connect the gadgets as follows.
	We add an edge between the last vertex of~$G_P$ and the first vertex~$s_{v_1}$ in~$G_{v_1}$ and an edge between~$t_{v_i}$ and~$s_{v_{i+1}}$ for all~$i \in \oneto{n-1}$.
	We then add an edge between~$t_{v_{n}}$ and~$s_{e_1}$ and between~$t_{e_i}$ and~$s_{e_{i+1}}$ for all~$i \in \oneto{m-1}$.
	Finally, we add an edge between~$t_{e_m}$ and each vertex in the first level of~$G_f$ and an edge between each vertex in the last level of~$G_f$ and~$t$.
		
	Note that since each gadget of~$G$ is an interval graph and since we only place them along a line,~$G$ is also an interval graph.
	Moreover,~$G$ can be constructed in polynomial time (since we may assume that~$k \leq n$).
	Thus, it only remains to be shown that~$G$ contains a \fair{} shortest~$s$-$t$-path if and only if~$H$ contains a vertex cover of size~$k$.
	
	Observe that any shortest path through a vertex gadget contains~${m+3}$~vertices, each shortest path through an edge gadget contains three vertices, each shortest path through~$G_f$ contains~$k \cdot (m+1) - m$ vertices, and the number of vertices in the (unique) path through~$G_P$ is
	$$km - k + km - 2(m+1) - 1 + n \cdot (km - (m+1)).$$
	Since the path ends in vertex~$t$ which is not part of any gadget, the number of vertices in each shortest~$s$-$t$-path in~$G$ is
	\begin{align*}
	& n \cdot (m+3) + 3m + k \cdot (m+1) - m \\
	&+ k m - k + k m - 2(m+n) + n \cdot (km - (m+1))\\
	= &\ nm + 3n + 3m + km + k - m + km - k + km - 2m - 2n +nkm -nm - n\\
	= &\ 3km + nkm = (n+3) \cdot km
	\end{align*}
	Since there are~$n+3$ colors, each color has to appear exactly~$k \cdot m$ times.
	
	For the forward direction, assume that there is a vertex cover~$S$ of size~$k$ in~$H$.
	We will construct an~$s$-$t$-path~$P$ in~$G$ which contains exactly~$k\cdot m$ vertices of each color.
	The path~$P$ contains all vertices in~$G_P$.
	For each vertex~$v \in S$, the path~$P$ contains in~$G_v$ all vertices in~$L_v$ and the two vertices~$s_v$ and~$t_v$.
	For each vertex~$u \in V \setminus S$, the path~$P$ contains in~$G_u$ all vertices in~$U_u$ and the two vertices~$s_u$ and~$t_u$.
	For each edge~$e \in E$, let~$v_e \in S$ be a vertex such that~$v_e \in e$.
	Note that~$v_e$ exists since~$S$ is a vertex cover in~$H$ and if both endpoints of~$e$ are contained in~$S$, then we choose an arbitrary endpoint.
	The path~$P$ contains for each edge~$e \in E$ the vertices~$s_e$,~$t_e$, and the vertex of color~$p_{v_e}$ in~$G_e$.
	For each vertex~$v \in S$, let~$x_v$ be the number of vertices of color~$p_v$ in~$P$ in the edge gadgets.
	The path~$P$ contains exactly~$(m+1)-x_v$ vertices of color~$p_v$ in~$G_f$.
	It is easy to verify that~$P$ contains exactly~$k \cdot m$ vertices of colors~$q_1$, $q_2$,~$r$, and of each color~$p_u$, where~$u \in V \setminus S$.
	So consider a color~$p_v$ with~$v \in S$.
	Observe that~$P$ contains~$k \cdot m - (m+1)$ vertices of color~$p_v$ in~$G_P$ and no vertex of that color in~$G_v$.
	Moreover, $P$ contains~$x_v$ vertices of that color in all of the edge gadgets combined and~$m+1 - x_v$ vertices of that color in~$G_f$.
	Hence, $P$~contains exactly~$k\cdot m$ vertices of color~$p_v$ and is therefore \fair.
	
	For the other direction, let~$P$ be a \fair{} shortest~\mbox{$s$-$t$-path} in~$G$.
	Since~$P$ passes through~$G_P$, it holds that $P$~contains exactly~$k$ vertices of color~$q_2$ outside of~$G_P$.
	Note that if~$P$ contains a vertex of color~$q_2$ in a vertex gadget~$G_v$, then it has to contain the whole path~$L_v$.
	Since there are exactly~$k$ vertex gadgets~$G_v$ such that~$P$ contains the whole path~$L_v$, it also holds that~$P$~contains exactly~$m \cdot k$ vertices of color~$q_1$ in these vertex gadgets.
	Hence, $P$~cannot contain any further vertices of color~$q_1$ and it therefore has to contain the whole path~$U_u$ for each other vertex gadget~$G_u$.
	Furthermore, for each vertex~$u \in V$ with~$P$ passing through~$U_u$, there cannot be any vertices of color~$p_u$ outside of~$G_P$ and~$G_u$.
	Hence, in order to pass through each edge gadget~$G_e$, there has to be a vertex~$v \in V$ such that~$P$ passes through~$L_v$ and~$v \in e$.
	Thus, these~$k$ vertices form a vertex cover in~$H$.
	This concludes the proof.
\end{proof}

Concluding this section, we next provide a reduction from \prob{Exact Cover} proving that \fsp{} is \NP-hard even on bipartite cactus graphs with constant bandwidth.
Recall that a graph is a \emph{cactus graph} if it is connected and any two cycles in the graph share at most one vertex, it is bipartite if its vertex set can be partitioned into two independent sets, and its bandwidth is the minimum cost~$\max_{\{u,v\}\in E}|f(u) - f(v)|$ over all injective functions~$f$ from~$V$ to~$\N$.
\prob{Exact Cover} is \NP-complete and defined as follows~\cite{Kar72}.

\problemdef{\prob{Exact Cover}}
{A universe~$U$ and a family~$\calF$ of subsets of~$U$.}
{Is there an \emph{exact cover} in~$\calF$, that is, is there a subfamily~${\calF' \subseteq \calF}$ such that for each~$x \in U$ there is exactly one~$S \in \calF'$ with~$x \in S$?}

\begin{proposition}
	\label{prop:bipartite}
	\fsp{} is \NP-hard even on bipartite cactus graphs with bandwidth two.
\end{proposition}
\begin{proof}
	We reduce from \prob{Exact Cover}.
	To this end, let~$I = ({U \defeq \oneto{n}}, \calF = \{S_1, S_2, \dots, S_m\})$ be an instance of \prob{Exact Cover} and let~${\sigma \defeq \sum_{S_j \in \calF} \abs{S_j}}$.
	We construct a bipartite cactus graph~$G = (V, E)$ with bandwidth two.
	Moreover, we define a coloring function~${\chi \colon V \to \oneto{n+1}}$.
	The graph~$G$ will contain two designated vertices~$s$ and~$t$ and there is a \fair{} shortest \stpath{} in~$G$ if and only if there is an exact cover in~$I$.
	Herein, the colors~$1, \dots, n$ represent the elements of the universe, and color~$n+1$ is used as a \emph{filler color}.
	The vertices~$s$ and~$t$ receive color~$n+1$.
	Additionally, the graph consists of a \emph{set gadget}~$G_j$ for each set~$S_j \in \calF$ and a path~$G_P$.

	For each~$S_j \in \calF$, let~${p_j = \abs{S_j}}$ and let~$\{x^j_1, x^j_2, \dots, x^j_{p_j}\} \subseteq \oneto{n}$ denote the elements of~$S_j$.
	The set gadget~$G_j$ consists of two vertices~$s_j$ and~$t_j$ and two parallel paths induced by the sets~$A_j \defeq \{a^j_1, a^j_2, \dots, a^j_{p_j}\}$ and~${B_j \defeq \{b^j_1, b^j_2, \dots, b^j_{p_j}\}}$, respectively.
	The endpoints are~$a^j_1$ and~$a^j_{p_j}$ and~$b^j_1$ and~$b^j_{p_j}$, respectively.
	The endpoints~$a^j_1$ and~$b^j_1$ are adjacent to~$s_j$, while~$a^j_{p_j}$ and~$b^j_{p_j}$ are adjacent to~$t_j$.
	The vertices~$s_j$ and~$t_j$ as well as the vertices in~$B_j$ are all colored with the filler color~$n+1$, and vertex~$a^j_q$ is colored with~$x^j_q$ for each~$q\in \oneto{p_j}$.

	The path~$G_P$ consists of~$n \cdot (2m + 1 + \sigma - n)$ vertices.
	For each color~$i \in \oneto{n}$, there are~$2m + 1 + \sigma - n$ vertices of color~$i$ in~$G_P$, and there are no vertices of color~$n+1$.
	We denote by~$V_P$ the set of vertices in~$G_P$.

	Finally, we connect the gadgets with the path and terminals by adding 
	an edge between~$s$ and one endpoint of~$G_P$,
	an edge between the other endpoint of~$G_P$ and~$s_1$,
	an edge between~$t_j$ and~$s_{j+1}$ for all~$j \in \oneto{m}$, 
	and an edge between~$t_m$ and~$t$.

	Observe that~$G$ is a bipartite cactus graph as the only cycles in~$G$ are the different set gadgets.
	These are cycles of even length that do not share vertices with any other cycles.
	We show that the bandwidth of our constructed graph~$G$ is at most~$2$ by constructing an injective function~$\varphi \colon V \to \N$ and proving~${\max_{\{u,v\}\in E}|\varphi(u) - \varphi(v)| \leq 2}$.
	We start by setting~$\varphi(s) = 1$ and assigning consecutive numbers to vertices adjacent in the path induced by~$\{s\} \cup V_P \cup \{s_1\}$.
	Suppose that we have assigned a number to~$s_j$ for some~$j \in \oneto{m}$.
	Then, we set~$\varphi(a^j_q) \defeq \varphi(s_j) + 2q-1$ and~$\varphi(b^j_q) \defeq \varphi(s_j) + 2q$ for each~$q \in \oneto{p_j}$.
	Next, we set~$\varphi(t_j) \defeq \varphi(b^j_p)+1$ and if~$j < m$, then we set~$\varphi(s_{j+1}) = \varphi(t_j)+1$.
	Lastly, we assign~$\varphi(t) = \varphi(t_m)+1$.
	By construction, all edges are between two vertices in a set gadget or their two endpoints are assigned consecutive numbers.
	Note that the endpoints of each edge within a set gadget~$G_j$ are assigned numbers with a difference of at most two (exactly two for all edges except for~$\{s_j,a^j_1\}$ and~$\{b^j_{p_j},t_j\}$).
	
	We next show that there is an exact cover in~$I$ if and only if~$G$ contains a \fair{} shortest \stpath{}.
	To this end, we first make some observations on the properties of any \stpath{} $P$ in~$G$.
	Each such path contains all vertices in
	\[
		C \defeq \{s, t\} \cup V_P \cup \bigcup_{S_j \in \calF} \{\{s_j, t_j\}\}.
	\]
	Note that~$2m+2$ of them are colored with color~$n+1$ and~${2m + 1 + \sigma - n}$ of them are colored with color~$i$ for each~$i \in \oneto{n}$. 
	For each~$j \in \oneto{m}$, the path~$P$ contains either the vertices in~$A_j$ or the vertices in~$B_j$.
	Let~$A_P$ be the set of vertices in~$P$ contained in~$A_j$ for some~$j \in \oneto{m}$ and let~$B_P$ be the set of vertices in~$P$ contained in~$B_j$ for some~$j \in \oneto{m}$.
	Then,~$\abs{A_P} + \abs{B_P} = \sigma$ and~$P$ contains overall~$2m+2 + \abs{B_P}$ vertices of color~$n+1$.
	
	Suppose that there is an exact cover~$\calF'$ in~$I$.
	Let~${A \defeq \bigcup_{S_j \in \calF'} A_j}$
	and let~${B \defeq \bigcup_{S_j \in \calF \setminus \calF'} B_j}$.
	We show that~$P \defeq G[A \cup B \cup C]$ induces a \fair{} shortest \stpath{} in~$G$.
	Since~$\calF'$ is an exact cover, it holds that~${\abs{A} = n}$ and~$A$~contains exactly one vertex of each color~$i \in \oneto{n}$.
	Hence, the number of vertices of each color in~$P$ is exactly~${2m + 2 + \sigma - n}$.
	Since the number of vertices of color~$n+1$ in~$P$ is~$2m+2 + \abs{B} = 2m+2 + \sigma - n$, it follows that~$P$ is \fair.

	For the converse, suppose that~$P$ is a \fair{} shortest \stpath{} in~$G$.
	Again, let~$A_P$ be the set of vertices in~$P$ that are contained in~$A_j$ for some~${j \in \oneto{m}}$ and let~$B_P$ be the set of vertices in~$P$ that are contained in~$B_j$ for some~$j \in \oneto{m}$.
	We claim that~$A_P$ contains exactly one vertex of each color~$i \in \oneto{n}$.
	Assume towards a contradiction that this is not the case.
	Then~$\abs{A_P} < n$,~$\abs{A_P} > n$, or~${\abs{A_P} = n}$ and at least one color appears at least twice in~$A_P$.
	In the first case, there exists a color~$i \in \oneto{n}$ that appears less than~$2m + 2 + \sigma - n$ times in~$P$.
	Recall that~$\abs{A_P} + \abs{B_P} = \sigma$ and therefore~${\abs{B_P} > \sigma - n}$.
	Hence, the color~$n+1$ appears more than $2m + 2 + \sigma - n$~times, a contradiction to~$P$~being \fair.
	In the second case, there exists a color~$i \in \oneto{n}$ that appears more than~${2m + 2 + \sigma - n}$~times in~$P$.
	Since~${\abs{B_P} < \sigma - n}$, the color~$n+1$ appears less than~$2m + 2 + \sigma - n$ times, a contradiction to~$P$ being \fair.
	In the third and final case, there exists a color~$i \in \oneto{n}$ that appears at least twice in~$A_P$.
	Then, there exists by the pigeonhole principle another color~$j \in \oneto{n}$ not appearing in~$A_p$ and the color~$i$ therefore appears at least two more times in~$P$ than color~$j$, a contradiction to~$P$ being \fair.
	Consequently,~${\calF' \defeq \{S_i \in \calF \mid A_i \subseteq A_P\}}$ is an exact cover in~$I$.
\end{proof}

Since cactus graphs are outerplanar, the above result also implies para-\NP-hardness for the parameters genus and distance to outerplanar graphs.

\section{Conclusion}
\label{sec:concl}
Our study shows that injecting \fair{ness} to the problem of finding shortest paths makes the problem computationally hard even when many (structural) graph parameters are small.
This is somewhat surprising as the computation of \fair{} spanning trees remains polynomial time solvable~\cite{brezovec1988matroid,CKLV19}.
A question that arises is how the tractability of other problems change when adding this natural fairness constraint.

In this work, we only focus on one fairness notion.
However, most of our results also hold for slight variations of \fair{ness}, \ie proportionality, giving appearance lower and upper bounds for each color.
Another fairness concept that could be studied with respect to paths is the margin of victory fairness that measures the difference between the most and second-most appearing color~\cite{boehmer2022fairmatching}.

\printbibliography
\end{document}